\documentclass[journal, onecolumn,12pt, draftclsnofoot]{IEEEtran}

\usepackage{listings}
\usepackage{cite}
\usepackage{url}
\usepackage{graphicx}
\usepackage{epstopdf}
\usepackage[cmex10]{amsmath}
\usepackage{amsfonts}
\usepackage{amssymb}
\usepackage{amsthm}
\theoremstyle{plain}
\newtheorem{theorem}{Theorem}

\usepackage[caption=false,font=normalsize,labelfont=sf,textfont=sf]{subfig}

\begin{document}
\title{Throughput Analysis of CSMA: Technical Report}
\author{Xinghua~Sun,~\IEEEmembership{Member,~IEEE}~and~Lin~Dai,~\IEEEmembership{Senior Member,~IEEE}
\thanks{X. Sun is with the School of Electronics and Communication Engineering, Sun Yat-sen University, China (email: sunxinghua@mail.sysu.edu.cn).}
\thanks{L. Dai is with the Department
of Electronic Engineering, City University of Hong Kong,
83 Tat Chee Avenue, Kowloon Tong, Hong Kong, China (email:  lindai@cityu.edu.hk).}
}

\maketitle

\begin{abstract}
In this technical report, the throughput performance of CSMA networks with two representative receiver structures, i.e., the collision model and the capture model, is characterized and optimized.  The analysis is further applied to an IEEE 802.11 network, which is a representative wireless network that adopts the CSMA mechanism, where the optimal initial backoff window sizes of nodes to achieve the maximum network throughput are derived and verified against simulation results.

\end{abstract}

\section{System Model and Preliminary Analysis}\label{I}
Consider a slotted CSMA network where $n$ nodes transmit to a common receiver.
Assume that each node always has packets in its buffer. With carrier sensing, each node should sense the channel first, and transmits only when the channel is idle. In this paper, nodes are assumed to be able to correctly sense the channel availability\footnote{Note that perfect sensing can be achieved even when nodes cannot hear each other's transmission. Specifically, as the receiver knows whether there are concurrent transmissions, it can broadcast this information so that all the nodes in the network can be informed of the channel availability. In the WiFi network with the request-to-send/clear-to-send (RTS/CTS) mechanism, for example, each WiFi node first sends an RTS frame, and then gets to know whether the channel is available by the CTS frame from the access point.}.
As Fig. \ref{agg} illustrates, the time axis of the aggregate channel is divided into multiple mini-slots with length $a$, where $a$ is the ratio of the propagation delay required by each node for sensing the channel to the packet length. The packet transmission lasts for one unit time, which is equal to ${1}/{a}$ mini-slots.
Assume that it takes each node $0 \le x \le 1/a$ mini-slots to know the failure of its transmitted packet\footnote{
Note that for wireless networks where nodes operate in the half-duplex mode, nodes are informed by the receiver about the outcome of their transmissions. How long for each node to be informed depends on the protocol design. In IEEE 802.11 DCF networks with the basic access mechanism, for instance, the receiver will send an ACK frame after the packet is successfully received. In this case, each node does not know whether the transmitted packet is successful or not until the end of the packet transmission. The failure-detection time $x$ is then determined by the packet length, i.e., $1/a$ mini-slots. On the other hand, if the RTS/CTS access mechanism is used, each node can send a short RTS frame to see whether its packet transmission can be successful or not by the response of the CTS frame from the receiver. In this case,  the failure-detection time $x$ is determined by the length of the RTS frame which is usually much shorter than the packet length, i.e., we have $x\ll1/a$.
} and abort the ongoing transmission.



\begin{figure}[htbp]
\centering
\includegraphics[width=120mm]{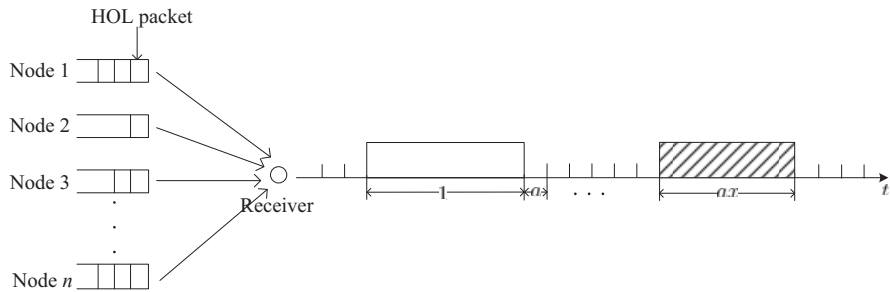}
\caption{Graphic illustration of an $n$-node CSMA network.}
\label{agg}
\end{figure}

For node $i$, the received power $P_{r,i}$ can be written as $P_{r,i} =P_{t,i} \cdot|g_i|^2\cdot|h_i|^2$,
where $P_{t,i}$ denotes the transmission power of node $i$, and $g_i$ and $h_i$ denote the large-scale and small-scale fading coefficients, respectively.
Assume block Rayleigh fading, i.e., $|h_i|^2\sim\exp(1)$ and $|h_i|^2$ varies from packet to packet,
and each node is aware of the large-scale fading coefficient by channel measurement, and thus can perform power control to combat the large-scale fading. In particular, each node sends packets with the transmission power $P_{t,i}=\frac{P}{|g_i|^2}$.
As a result, the mean received power is the same for each node. The mean received SNR can then be written as $\rho=P/\sigma^2$, where $\sigma^2$ denotes the noise power.

\vspace{-0.5cm}
\subsection{Transmitter Model}
It has been shown in \cite{CSMA_Aloha} that the performance of CSMA networks is crucially determined by activities of HOL packets of nodes' queues. The behavior of each HOL packet in each node's queue can be characterized by a discrete-time Markov renewal process $(\textbf{\textit{X}},\textbf{\textit{V}})=\{(X_j,V_j), j=0,1,\ldots\}$, where $X_j$  denotes the state of one HOL packet at the $j$-th transition and $V_j$ denotes the epoch at which the $j$-th transition occurs.

As Fig. \ref{Chain} illustrates, the states of $\{X_j\}$ can be divided into three categories: 1) waiting to request (State $\text{R}_i$, $i=0,\ldots,K$), 2) failure (State
$\text{F}_i$, $i=0,\ldots,K$) and 3) successful transmission (State $\text{T}$).
A State-$\text{R}_i$ HOL packet has a transmission probability of $q_i$, $i=0,\ldots,K$, at each idle mini-slot. It moves to State $\text{T}$ if it transmits and the transmission
is successful. Otherwise, if the transmission fails, it moves to State $\text{F}_i$ and then shifts to State $\text{R}_{i+1}$.
If the HOL packet has experienced more than $K$ transmission failures, its transmission probability remains to be $q_K$.
Here $K$ is referred to as the cutoff phase.  To alleviate channel contention, $\{q_i\}_{i=0,\ldots,K}$ is usually assumed to be a monotonic non-increasing sequence. Without loss of generality, let $q_i = q_0 \cdot \mathcal{Q}(i)$, where $q_0$ is the initial transmission probability and $\mathcal{Q}(i)$ is an arbitrary monotonic non-increasing function of $i$ with $\mathcal{Q}(0) = 1$ and $\mathcal{Q}(i)\le\mathcal{Q}(i-1)$, $i = 1,\ldots,K$.
\begin{figure}[htbp]
\centering
\includegraphics[width=100mm]{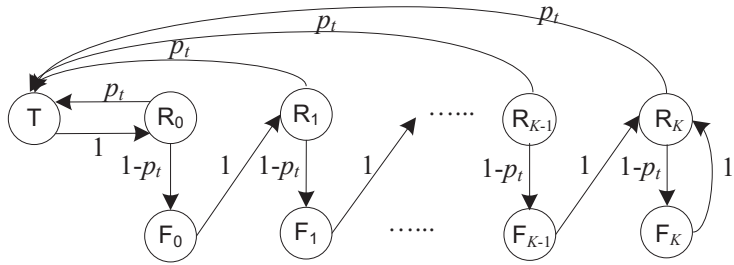}
\caption{Embedded Markov chain \{$X_j$\} of the state transition
process of each HOL packet in CSMA networks.}
\label{Chain}
\end{figure}

In Fig. \ref{Chain}, $p_t$ represents the probability of successful transmission of HOL packets at mini-slot $t$ given that the channel is idle at mini-slot $t-1$.
Let $\{\tilde{\pi }^{}_{i}\}$ denote the limiting state probabilities of the Markov renewal process. We then have
\begin{equation} \label{ls_def}
\tilde{\pi }^{}_{i} =\frac{\pi^{} _{i} \cdot \tau ^{}_{i} }{\sum _{j\in
S}\pi ^{}_{j} \cdot \tau ^{}_{j}  } ,
\end{equation}
$i\in S$, where $S=\{\text{T},\text{F}_0,\ldots,\text{F}_K,\text{R}_0,\ldots,\text{R}_K\}$ is the state space of $\textbf{\textit{X}}^{}$, $\{\pi_i\}_{i\in S}$ denotes the steady-state probability distribution of the
embedded Markov chain, and $\tau_i$ denotes the mean holding time in each state ${i\in S}$.
Specifically, the probability of being in State $\textmd{T}$ for the HOL packet, $\tilde{\pi }_{T}^{}$, has been derived in \cite{CSMA_Aloha} as
\begin{align}\label{sr}
\tilde{\pi}_{T}=\frac{1}{1+xa\cdot\frac{1-p}{p}{+}
\frac{a}{\alpha p}\cdot\left(\sum
_{i=0}^{K{-}1}\frac{p(1{-}p)^{i}}{q_i}{+}\frac{(1{-}p)^K}{q_K}\right)},
\end{align}
where $p=\lim_{t\to\infty}p_t$ is the steady-state probability of successful transmission of HOL packets given that the channel is idle, $\alpha$ denotes the steady-state probability of sensing the channel idle, and $q_i$ is the transmission probability of a HOL packet in State $\text{R}_i$ given that the channel is idle. Note that $\tilde{\pi}_{T}$ is the service rate of each node's queue as the queue has a successful output if and only if the HOL packet is in State $\text{T}$.

Similar to \cite{Aloha_cap}, it is assumed that
the transmitters are unaware of the
instantaneous realizations of the small-scale fading coefficients. As a result, each
transmitter independently encodes its information at a constant rate $R$ bit/s/Hz. The network sum rate $R_s$, which is defined as the average received information rate, can be written as \cite{Aloha_cap}
\begin{align} \label{Rs}
R_s=\hat{\lambda}_{\text{out}}\cdot R,
\end{align}
where the network throughput $\hat{\lambda}_{\text{out}}$ is the average number of successfully decoded packets per time slot, which depends on the transmission probabilities $\{q_i\}_{i=0,\dots,K}$ of each node and the receiver structure.


\subsection{Receiver Model}\label{II-B}
In the literature, different receiver structures have been proposed, among which the collision model and the capture model are two representative ones. With the classic collision model, one packet can be successfully decoded only if there are no concurrent transmissions.
Although the collision model greatly simplifies the analysis, it leads to pessimistic evaluation of the network performance. In practice, one packet could be successfully ``captured" even with multiple concurrent transmissions, as long as the received power is sufficiently high compared to that of the interference. With the capture model,
each node's packet is decoded independently by treating
others¡¯ as background noise.
In this paper, we consider the above two receiver structures:

1) Collision model:  one packet can be successfully decoded if and only if
there are no concurrent transmissions and its received SNR is above a certain threshold;

2) Capture model: one packet can be successfully
decoded as long as its received SINR is above a certain threshold.

Let
\begin{equation}\label{rate}
\mu=2^R-1
\end{equation}
denote the threshold at the receiver. For each node's
packet, if its received SNR (SINR) exceeds the receiver threshold $\mu$ with the collision (capture) model, then by random coding the error probability is
exponentially reduced to zero as the packet length goes to infinity. In this paper, we assume that the packet length is sufficiently large such that the rate $R$ can be supported for reliable communications\footnote{Note that the error probability may become non-negligible when the packet length is not sufficiently large. In that case, the receiver threshold $\mu$ could be dependent on not only the information encoding rate $R$, but also the error probability that is determined by the packet length and the coding/decoding schemes.}.



In the following, we will characterize the sum rate performance of CSMA networks under the above two receiver structures. For differentiation purpose, performance metrics are marked with superscript ``Col" for the collision model and ``Cap" for the capture model, respectively.

\section{Network Throughput}\label{II}
%

In saturated conditions, the throughput of each node is equal to the service rate of each node's queue. The network throughput $\hat{\lambda}^{}_{\text{out}}$ can then be written as
\begin{equation}\label{th_def}
\hat{\lambda}^{}_{\text{out}}=n\tilde{\pi }_{T }=\frac{n}{1+xa\cdot\frac{1-p}{p}{+}
\frac{a}{\alpha p}\cdot\left(\sum
_{i=0}^{K{-}1}\frac{p(1{-}p)^{i}}{q_i}{+}\frac{(1{-}p)^K}{q_K}\right)},
\end{equation}
where $\tilde{\pi }_{T }$ is the probability of being in State
$\textmd{T}$ for the HOL packet, which is given in \eqref{sr}.
It can be seen from \eqref{th_def} that the network throughput $\hat{\lambda}^{}_{\text{out}}$ critically depends on the steady-state probability of successful transmission of HOL packets given that the channel is idle, $p$.  In the following, we will first characterize the network steady-state point in saturated conditions based on the fixed-point equation of $p$, and then obtain the maximum network throughput for both the collision and capture models.

\subsubsection{Steady-state Point in Saturated Conditions}\label{III-A-1}
It is shown in Appendix \ref{app0} that the network steady-state point $p^{\text{Col}}_A$ for the collision model is the single non-zero root of the following fixed-point equation
\begin{align} \label{Eq7}
&p^{\text{Col}}=\exp\left\{{-}\frac{\mu}{\rho}\right\}{\cdot}\exp \left\{{-}\frac{ n}{\frac{\alpha^{\text{Col}}}{a}{+}\sum
_{i=0}^{K{-}1}\frac{p^{\text{Col}}(1{-}p^{\text{Col}})^{i}}{q_i}{+}\frac{(1{-}p^{\text{Col}})^K}{q_K}} \right\}\notag\\
&\mathop{\approx }\limits^{{\rm with \;large}\; K}\exp\left\{-\frac{\mu}{\rho}\right\}\cdot \exp \left\{-\frac{ n}{\sum
_{i=0}^{K{-}1}\frac{p^{\text{Col}}(1{-}p^{\text{Col}})^{i}}{q_i}{+}\frac{(1{-}p^{\text{Col}})^K}{q_K}} \right\},
\end{align}
where the probability of sensing the channel idle $\alpha^{\text{Col}}$ is given by
\begin{align} \label{Eq9}
\alpha^{\text{Col}}=\frac{a}{(x+1)a-(1-ax)p^{\text{Col}}\left(\frac{\mu}{ \rho}+\ln p^{\text{Col}}\right)-ax\exp\left\{\frac{\mu}{ \rho}\right\}p^{\text{Col}}}.
\end{align}
On the other hand, for the capture model, the network steady-state point $p^{\text{Cap}}_A$ is the single non-zero root of the following fixed-point equation
\begin{align} \label{pA_app}
&p^{\text{Cap}}=\exp\left\{{-}\frac{\mu}{\rho}\right\}{\cdot}\exp \left\{{-}\frac{\frac{\mu}{1+\mu}\cdot n}{\frac{\alpha^{\text{Cap}}}{a}{+}\sum
_{i=0}^{K{-}1}\frac{p^{\text{Cap}}(1{-}p^{\text{Cap}})^{i}}{q_i}{+}\frac{(1{-}p^{\text{Cap}})^K}{q_K}} \right\}\notag\\
&\mathop{\approx }\limits^{{\rm with \;large}\; K}\exp\left\{-\frac{\mu}{\rho}\right\}\cdot \exp \left\{-\frac{\frac{\mu}{1+\mu}\cdot n}{\sum
_{i=0}^{K{-}1}\frac{p^{\text{Cap}}(1{-}p^{\text{Cap}})^{i}}{q_i}{+}\frac{(1{-}p^{\text{Cap}})^K}{q_K}} \right\},
\end{align}
where the probability of sensing the channel idle $\alpha^{\text{Cap}}$ can be obtained as
\begin{align} \label{alfa}
&\alpha^{\text{Cap}}={a}/\Bigg(1{+}a{-}\exp\left\{\frac{1{+}\mu}{\rho}\right\}(p^{\text{Cap}})^{\frac{1+\mu}{\mu}}{-}(1{-}ax)
\sum_{i=1}^n\left(1{-}\frac{\exp\{{-}\frac{\mu}{\rho}\}}{(1{+}\mu)^{i-1}}\right)^i\binom{n}{i}\notag\\
&\cdot\left(\left({-}\ln p^{\text{Cap}}{-}\frac{\mu}{\rho}\right){\cdot} \frac{1{+}\mu}{n\mu}\right)^i\left(1{-}\left({-}\ln p^{\text{Cap}}{-}\frac{\mu}{\rho}\right){\cdot} \frac{1{+}\mu}{n\mu}\right)^{n{-}i}\Bigg).
\end{align}

It is indicated in \eqref{Eq7} and \eqref{pA_app} that both $p^{\text{Col}}_A$ and $p^{\text{Cap}}_A$ are determined by the network size $n$, the receiver threshold $\mu$, the mean received SNR $\rho$ and the sequence of transmission probabilities $\{q_i\}_{i=0,\ldots,K}$. For the collision model, \eqref{Eq7} reduces to Eq. (51) in \cite{CSMA_Aloha} if $\rho\to\infty$. In this case, one packet can be successfully received as long as there are no concurrent transmissions.

\subsubsection{Maximum Network Throughput}\label{}
By combining  \eqref{th_def}, \eqref{Eq7} and \eqref{Eq9}, the network throughput with the collision model, $\hat{\lambda}^{\text{Col}}_{\text{out}}$, can be obtained as
\begin{equation}\label{Eq11}
\hat{\lambda}^{\text{Col}}_{\text{out}}=
\frac{\frac{1}{ax}}{\frac{1+\frac{1}{x}-\exp\left\{\frac{\mu}{ \rho}\right\}\cdot p^{\text{Col}}_A}{-p^{\text{Col}}_A\left(\frac{\mu}{ \rho}+\ln p^{\text{Col}}_A\right)}+\frac{1}{ax}-1}.
\end{equation}
With the capture model, the network throughput $\hat{\lambda}^{\text{Cap}}_{\text{out}}$ is given by
\begin{align} \label{th}
&\hat{\lambda}^{\text{Cap}}_{\text{out}}={-p_A^{\text{Cap}}\left(\frac{1+\mu}{ \rho}+\frac{1+\mu}{\mu}\ln p_A^{\text{Cap}}\right)}/\Bigg(1+a-\exp\left\{\frac{1+\mu}{\rho}\right\}(p_A^{\text{Cap}})^{\frac{1+\mu}{\mu}}-(1-ax)\notag\\
&\sum_{i=1}^n\left(1{-}\frac{\exp\left\{{-}\frac{\mu}{\rho}\right\}}{(1{+}\mu)^{i-1}}\right)^i\binom{n}{i}\left(\left({-}\ln p_A^{\text{Cap}}{-}\frac{\mu}{\rho}\right){\cdot} \frac{1{+}\mu}{n\mu}\right)^i\left(1{-}\left({-}\ln p_A^{\text{Cap}}{-}\frac{\mu}{\rho}\right){\cdot} \frac{1{+}\mu}{n\mu}\right)^{n-i}\Bigg),
\end{align}
by combining \eqref{th_def}, \eqref{pA_app} and \eqref{alfa}.

For given mini-slot length $a$, the
failure-detection time $x$, the receiver threshold $\mu$ and the mean received SNR $\rho$, both $\hat{\lambda}^{\text{Col}}_{\text{out}}$ and $\hat{\lambda}^{\text{Cap}}_{\text{out}}$ are functions of the network steady-state point according to \eqref{Eq11} and \eqref{th}, which in turn are determined by backoff parameters $\{q_i\}$ according to \eqref{Eq7} and \eqref{pA_app}. 
To maximize the network throughput, the backoff parameters $\{q_i\}$ should be carefully tuned.
The following theorems present the maximum network throughput and the corresponding optimal backoff parameters for the collision model and the capture model, respectively.


\begin{theorem}\label{t1}
With the collision model, the maximum network throughput $\hat{\lambda}^{\text{Col}}_{\max}=\max_{\{q_i\}}\hat{\lambda}^{\text{Col}}_{\text{out}}$ is given by
\begin{equation}\label{th_max_col}
\hat{\lambda}^{\text{Col}}_{\max}=
\frac{{-{\it {\mathbb W}}_{0}\left(-\tfrac{1}{e(1+1/x)}\right)}}{\exp\{\frac{\mu^{}}{\rho}\}\cdot ax-(1-ax){{\it {\mathbb W}}_{0}\left(-\tfrac{1}{e(1+1/x)}\right)}},
\end{equation}
where ${\it {\mathbb W}}_{0}(\cdot)$ is the principal branch of the
Lambert W function \cite{Corlessa}. $\hat{\lambda}^{\text{Col}}_{\max}$ is achieved when
\begin{equation}\label{temp}
q^{*,{\text{Col}}}_i=\hat{q}^{\text{Col}}_0\mathcal{Q}(i),
\end{equation}
$i=0,\ldots,K$, where $\hat{q}^{\text{Col}}_0$ is given by
\begin{equation}\label{qm_col}
\hat{q}^{\text{Col}}_0=-\frac{\ln \psi^{*,{\text{Col}}}}{n}\left({\sum
_{i=0}^{K{-}1}\frac{e^{-\frac{\mu}{\rho}}\psi^{*,{\text{Col}}}\left(1{-}e^{-\frac{\mu}{\rho}}\psi^{*,{\text{Col}}}\right)^{i}}{\mathcal{Q}(i)}{+}\frac{\left(1-e^{-\frac{\mu}{\rho}}\psi^{*,{\text{Col}}}\right)^K}{\mathcal{Q}(K)}}\right),
\end{equation}
with $\psi^{*,{\text{Col}}}={-(1+1/x){\it {\mathbb W}}_{0}\left(-\tfrac{1}{e(1+1/x)}\right)}$.
\end{theorem}
\begin{proof}
See Appendix \ref{app1}.
\end{proof}

\begin{theorem}\label{t3}
With the capture model, the maximum network throughput $\hat{\lambda}^{\text{Cap}}_{\max}=\max_{\{q_i\}}\hat{\lambda}^{\text{Cap}}_{\text{out}}$ is given by
\begin{equation}\label{th_max}
\hat{\lambda}^{\text{Cap}}_{\max}=\left\{\!\!\!
\begin{array}{ll}
\frac{-\exp\left\{-\frac{\mu}{\rho}\right\}(\psi^{*,{\text{Cap}}})^{\frac{\mu}{1+\mu}}\ln \psi^{*,{\text{Cap}}}}{1+a-\psi^{*,{\text{Cap}}}-(1-ax)\sum_{i=1}^n\left(1-\frac{\exp\{-\frac{\mu}{\rho}\}}{(1+\mu)^{i-1}}\right)^i\binom{n}{i}\left(-\frac{\ln\psi^{*,{\text{Cap}}}}{n}\right)^i\left(1+\frac{\ln\psi^{*,{\text{Cap}}}}{n}\right)^{n-i}} \,\,\,\,\,\,\,\,\,\,\,\,\,\,\,\text{if}\,\,\,\,\,\, \mu\ge \mu^{}_0\\
\frac{n\exp\left\{-\frac{\mu_{}}{\rho}\right\}\exp\left\{-\frac{n\mu_{}}{1+\mu_{}}\right\} }{1+a-\exp\{-n\}-(1-ax)\left(1-\frac{\exp\left\{-\frac{\mu_{}}{\rho}\right\}}{(1+\mu)^{n-1}}   \right)^n}\,\,\,\,\,\,\,\,\,\,\,\,\,\,\,\,\,\,\,\,\,\,\,\,\,\,\,\,\,\,\,\,\,\,\,\,\,\,\,\,\,\,\,\,\,\,\,\,\,\,\,\,\,\,\,\,\,\,\,\,\,\,\,\,\,\,\,\,\,\,\,\,\,\,\,\,\,\,\,\,\,\,\,\,\text{otherwise,}\,\,\,\,\,\,
\end{array}\right.
\end{equation}
where
\begin{equation}\label{mu0}
\mu^{}_0=\frac{1+a-\exp\{-n\}}{\left(1{-}\frac{1}{n}\right)(1{+}a{-}\exp\{-n\}){+}\exp\{{-}n\}{-}(1{-}ax)\left(
\left(1{-}\frac{\exp\left\{{-}\frac{\mu_{}}{\rho}\right\}}{(1{+}\mu)^{n-1}}   \right)^n{-}\left(1{-}\frac{\exp\left\{{-}\frac{\mu_{}}{\rho}\right\}}{(1{+}\mu)^{n-2}}   \right)^{n-1}\right)}-1
\end{equation}
and
$\psi^{*,\text{Cap}}$ is the root of
\begin{align}\label{psi}
&\left(\frac{\mu}{1+\mu}\ln \psi+1\right)\cdot\left(1+a-\psi\right)+\psi\ln \psi-(1-ax)\sum_{i=1}^n\left(1-\frac{\exp\left\{-\frac{\mu}{\rho}\right\}}{({1+\mu})^{i-1}}\right)^i\binom{n}{i}\left(-\frac{\ln\psi}{n}\right)^i\notag\\
&\left(1+\frac{\ln\psi}{n}\right)^{n-i-1}\left( \left(\frac{\mu}{1+\mu}\ln \psi+1\right)\left(1+\frac{\ln\psi}{n}\right)-i-\ln\psi\right)=0.
\end{align}
$\hat{\lambda}^{\text{Cap}}_{\max}$ is achieved when
\begin{equation}\label{qm_cap}
q^{*,\text{Cap}}_i=\left\{\!\!\!
\begin{array}{ll}
\hat{q}^{\text{Cap}}_0\mathcal{Q}(i)\,\,\,\,\,\,\,\,\,\,\,\,\,\text{if}\,\,\,\,\,\, \mu\ge \mu^{}_0\\
1 \,\,\,\,\,\,\,\,\,\,\,\,\,\,\,\,\,\,\,\,\,\,\,\,\,\,\,\,\,\,\,\text{otherwise,}\,\,\,\,\,\,
\end{array}\right.
\end{equation}
$i=0,\ldots,K$, where $\hat{q}^{\text{Cap}}_0$ is given by
\begin{equation}\label{temp2}
\hat{q}^{\text{Cap}}_0=-\frac{\ln \psi^{*,\text{Cap}}}{n}\left({\sum
_{i=0}^{K{-}1}\frac{e^{-\frac{\mu}{\rho}}(\psi^{*,\text{Cap}})^{\frac{\mu}{1+\mu}}\left(1{-}e^{-\frac{\mu}{\rho}}(\psi^{*,\text{Cap}})^{\frac{\mu}{1+\mu}}\right)^{i}}{\mathcal{Q}(i)}{+}\frac{\left(1-e^{-\frac{\mu}{\rho}}(\psi^{*,\text{Cap}})^{\frac{\mu}{1+\mu}}\right)^K}{\mathcal{Q}(K)}}\right).
\end{equation}
\end{theorem}
\begin{proof}
See Appendix \ref{app3}.
\end{proof}

It is clear from \eqref{th_max_col} and \eqref{th_max} that both $\hat{\lambda}^{\text{Col}}_{\max}$ and $\hat{\lambda}^{\text{Cap}}_{\max}$  depend on the mini-slot length $a$, the failure-detection time $x$, the receiver threshold $\mu$ and the mean received SNR $\rho$.
For the collision model, \eqref{th_max_col} reduces to the maximum network throughput in perfect channel conditions, i.e., Eq. (9) in \cite{CSMA_Aloha}, as $\rho\to\infty$. For the capture model, it can be shown that
\begin{align}\label{mu_app}
\mu^{}_0\mathop{\approx }\limits^{{\rm with \;large}\; n} \frac{1}{n-1}
\end{align}
and
\begin{align}\label{thmax_low_app}
\hat{\lambda}^{\text{Cap},\mu\le\mu^{}_0}_{\max}\mathop{\approx }\limits^{{\rm with \;large}\; n}\frac{n}{1+a}\exp\left\{-\frac{\mu_{}}{\rho}\right\}\exp\left\{-\frac{n\mu_{}}{1+\mu_{}}\right\}.
\end{align}
It can be seen from \eqref{thmax_low_app} that $\hat{\lambda}^{\text{Cap},\mu\le\mu^{}_0}_{\max}$ becomes insensitive to the failure-detection time $x$ when the number of nodes $n$ is large.
As shown in \eqref{mu_app}, with a large $n$, $\mu^{}_0\ll 1$. With such a small threshold, each packet has a high probability of being successfully decoded, and thus the probability
that a transmission failure occurs, i.e., all of concurrently-transmitted packets fail,  becomes close to zero.

\begin{figure}[htbp]
\centering
\includegraphics[width=3.3in,height=2.5in]{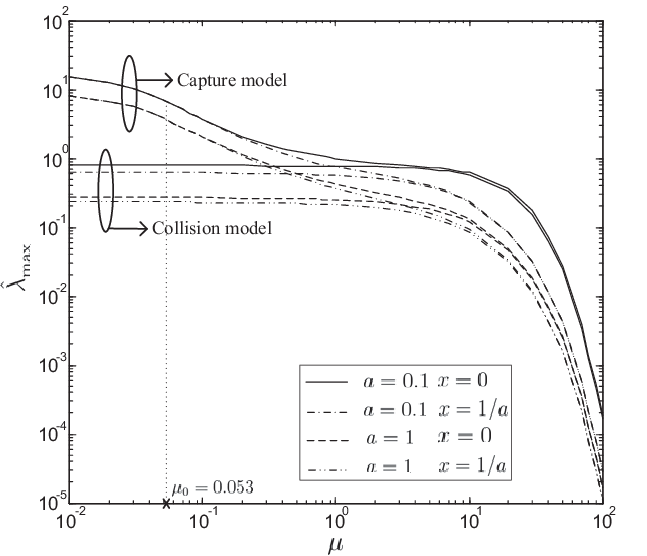}
\caption{Maximum network throughput $\hat{\lambda}_{\max}$ versus receiver threshold $\mu$. $\rho=10\text{dB}$ and $n=20$.}
\label{th_mu}
\end{figure}

Fig. \ref{th_mu} demonstrates how the maximum network throughput $\hat{\lambda}_{\max}$ varies with the receiver threshold $\mu$ for both the collision and capture models. Intuitively, fewer packets can be successfully decoded if the receiver threshold $\mu$ is enlarged. Therefore,
it can be seen from Fig. \ref{th_mu} that both $\hat{\lambda}^{\text{Col}}_{\max}$ and $\hat{\lambda}^{\text{Cap}}_{\max}$ decrease as $\mu$ increases, and the gain of $\hat{\lambda}^{\text{Cap}}_{\max}$ over $\hat{\lambda}^{\text{Col}}_{\max}$ disappears when $\mu$ is sufficiently large, in which case the capture model reduces to the collision model as only one packet can be successfully decoded each time. Moreover, it can be observed that both $\hat{\lambda}^{\text{Col}}_{\max}$ and $\hat{\lambda}^{\text{Cap}}_{\max}$ increase as the mini-slot length $a$ decreases. This is because with a smaller $a$, the channel contention can be distributed over time in a more refined manner, leading to lower chances of transmission failures. As the channel time wasted in transmission failures is reduced with a smaller failure-detection time $x$, $\hat{\lambda}^{\text{Col}}_{\max}$ and $\hat{\lambda}^{\text{Cap}}_{\max}$ can also be improved as $x$ decreases, as Fig. \ref{th_mu} illustrates.

\section{Case Study: IEEE 802.11 DCF Networks}

%
In Section \ref{II}, we have obtained an explicit expression of the network throughput, and demonstrated how to maximize the network throughput for a general CSMA network. The CSMA mechanism has been widely adopted in various types of practical networks, among which the IEEE 802.11 Distributed Coordination Function (DCF) network is a typical example. In this section, we will elaborate on how the above analysis can be applied to an IEEE 802.11 DCF network, and then validate the analysis using the ns-2 simulator by taking the protocol details of DCF into consideration.


\begin{figure*}[ht]
    \centering
        \includegraphics[width=0.9\textwidth]{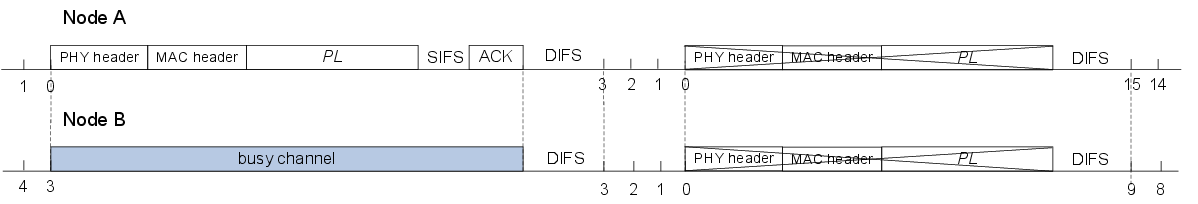}
\caption{Graphic illustration of the transmission behavior of nodes in IEEE 802.11 DCF networks with the basic access mechanism.}\label{figure1}
\end{figure*}

\subsection{Throughput Analysis}
Recall that for a general CSMA network, the key parameters include 1) the mini-slot length $a$, which is the ratio of the propagation delay required by each node for sensing the channel to the packet length, 2) the failure-detection time $x$, which is the time each node needs to know the failure of its transmitted packet and abort the ongoing transmission, and 3) the sequence of transmission probabilities $\{q_i\}_{i=0,\ldots,K}$, which means that each node has a transmission probability of $q_i$, ${i=0,\ldots,K}$, at each idle mini-slot after the $i$th transmission failure. In the following, we will demonstrate how to map these key parameters into those of an IEEE 802.11 DCF network.

Specifically, for each node, the state transition process of its head-of-line (HOL) packet has been established in Section III, where a HOL packet may stay in 1) the successful transmission state, i.e., State $T$, 2) the failure states, i.e., State $\textmd{F}_i$, $i=0,\ldots,K$, or 3) the waiting states, i.e., State $\textmd{R}_i$, $i= 0,\ldots,K$. Let $\tau_T$ and $\tau_F$ denote the holding time of a HOL packet in State $T$ and State $\textmd{F}_i$, respectively, in unit of mini slots. We then have
\begin{equation}
a=1/\tau_T\label{eq1}
\end{equation}
and
\begin{equation}
x=\tau_F.\label{eq2}
\end{equation}

\begin{table}[h]
\caption{Parameter Setting}
\centering
\begin{center}
\begin{tabular}{|c|c|}
\hline
packet payload ($PL$) & 2048B\\
PHY header & 20$\mu$s\\
MAC header & 36B\\
ACK & 14B+PHY header \\ 
\hline 
Slot time ($\sigma$) & 9$\mu$s\\
SIFS & 16$\mu$s\\
DIFS & 34$\mu$s\\
Basic rate & 6Mb/s\\
Transmission rate & 65Mb/s\\
\hline
\end{tabular}\label{table1}
\end{center}
\end{table}

In IEEE 802.11 DCF networks, the holding time of a HOL packet in State $T$ and State $\textmd{F}_i$, $\tau_T$ and $\tau_F$, vary with different access mechanisms. Fig. \ref{figure1} demonstrates a simple example of the transmission behavior of nodes in IEEE 802.11 DCF networks with the basic access mechanism. According to Fig. \ref{figure1}, $\tau_T$ and $\tau_F$ (in unit of mini slots) can be written as
\begin{align}\label{tauT}
\tau_T^{}=\frac{(PL+\text{MAC header})/r^{}}{\sigma}+\frac{\text{PHY header}{+}\text{ACK}/\text{Basic rate}{+}\text{DIFS}{+}\text{SIFS}}{\sigma}
\end{align}
and
\begin{align}\label{tauF}
\tau_F^{}=\frac{(PL+\text{MAC header})/r^{}}{\sigma}+\frac{\text{PHY header}{+}\text{DIFS}}{\sigma},
\end{align}
respectively, where SIFS and DIFS are abbreviations for Short Interframe Space and DCF Interframe Space, respectively, and $r$ is the transmission rate of each node. With the system parameters adopted in the IEEE 802.11n standard \cite{IEEE_Std2009}, which are provided in Table \ref{table1}, for example, we can obtain from (\ref{tauT}-\ref{tauF}) that $\tau_T= 40.44$ mini slots and $\tau_F= 34.36$ mini slots. As a result, we have in this case $a=1/\tau_T=0.0247$ and $x=\tau_F=34.36$.

The holding time of a HOL packet in State $\textmd{R}_{i}$, $i=0,\dots,K$, on the other hand, is determined by the backoff protocol. In IEEE 802.11 DCF networks, when a HOL packet enters State $\textmd{R}_{i}$, it randomly selects a value from $\{0,\dots,W_{i}-1\}$, where $W^{}_{i}=W^{}_{}\cdot 2^{i}$ is the backoff window size at State $\textmd{R}_{i}$, $i=0,\dots,K$, and $W$ is the initial backoff window size of each node. The HOL packet then counts down at each idle time slot. It leaves State $\textmd{R}_{i}$ and makes a transmission request when the channel is idle and the counter is zero. The mean holding time at State $\textmd{R}_{i}$, $i=0,\dots,K$, can then be obtained as \cite{Unified}
\begin{equation}\label{tau_W}
\tau^{}_{{R}_{i}}=\frac{1}{\alpha}\cdot \frac{1+W^{}_i}{2},
\end{equation}
where $\alpha$ is the probability of sensing the channel idle. As a node with a State-$\textmd{R}_{i}$ HOL packet would access the channel with the transmission probability $q^{}_i$ when it senses the channel idle, the mean holding time $\tau _{R_{i} }$ (in unit of mini slots) can be obtained according to Appendix B in \cite{CSMA_Aloha} as
\begin{equation} \label{tau_R}
\tau^{} _{R_{i} }=\frac{1}{\alpha q^{}_i}.
\end{equation}
By combining (\ref{tau_W}) and (\ref{tau_R}), the transmission probability of a State-$\textmd{R}_{i}$ HOL packet when it senses the channel idle can be written as
\begin{equation}\label{trans}
q^{}_i=\frac{2}{1+W^{}_i},
\end{equation}
$i=0,\dots,K$. (\ref{trans}) indicates that with a larger backoff window size, each node has a smaller chance to access the channel.

The network steady-state points of a saturated CSMA network have been derived as \eqref{Eq7} and \eqref{pA_app} for the collision and capture models, respectively. By combining (\ref{trans}) with \eqref{Eq7}, the network steady-state point of a saturated IEEE 802.11 DCF network with the collision model can be obtained as the single non-zero root of
\begin{align}
p^{\text{Col}}&=\exp\left\{-\frac{\mu}{\rho}\right\}\cdot\exp \Bigg\{-\frac{2n^{}}{1+\sum
_{i=0}^{K-1}{p^{\text{Col}}(1{-}p^{\text{Col}})^{i}}{W^{}_i} {+}{(1{-}p^{\text{Col}})^{K}}{W^{}_K}}\Bigg\}.\label{pAW1}
\end{align}
With the capture model, on the other hand, the network steady-state point of a saturated IEEE 802.11 DCF network can be obtained by combining (\ref{trans}) with \eqref{pA_app} as the single non-zero root of
\begin{align}
p^{\text{Cap}}&=\exp\left\{-\frac{\mu}{\rho}\right\}\cdot\exp \Bigg\{-\frac{\frac{2\mu}{1+\mu}\cdot n}{1+\sum
_{i=0}^{K-1}{p^{\text{Cap}}(1{-}p^{\text{Cap}})^{i}}{W^{}_i} {+}{(1{-}p^{\text{Cap}})^{K}}{W^{}_K}}\Bigg\}.\label{pAW2}
\end{align}
The network throughput with the collision model $\hat{\lambda}^{\text{Col}}_{\text{out}}$ and the network throughput with the capture model $\hat{\lambda}^{\text{Cap}}_{\text{out}}$ can then be obtained by substituting (\ref{tauT}-\ref{tauF}) and (\ref{pAW1}) into \eqref{Eq11}, and substituting (\ref{tauT}-\ref{tauF}) and (\ref{pAW2}) into \eqref{th}, respectively.

Moreover, the maximum network throughputs with the collision and capture models have been shown in Theorem 1 and Theorem 2, respectively. To achieve the maximum network throughput, the initial transmission probability $q_0$ needs to be carefully tuned. For an IEEE 802.11 DCF network, the maximum network throughputs with the collision and capture models can be derived by substituting (\ref{tauT}-\ref{tauF}) into \eqref{th_max_col} and substituting (\ref{tauT}-\ref{tauF}) into \eqref{th_max}, respectively. To achieve the maximum network throughput, the initial backoff window sizes should be properly tuned. With the collision model, the optimal backoff window sizes can be obtained by combining (\ref{trans}) and (\ref{temp}-\ref{qm_col}) as
\begin{equation}
W^{*,{\text{Col}}}_i=\hat{W}^{\text{Col}}_0 2^i,\notag
\end{equation}
$i=0,\ldots,K$, where
\begin{equation}\label{temp3}
\hat{W}^{\text{Col}}_0=\frac{-2n}{\ln \psi^{*,{\text{Col}}}}\left( \frac{e^{-\frac{\mu}{\rho}}\psi^{*,{\text{Col}}}}{2e^{-\frac{\mu}{\rho}}\psi^{*,{\text{Col}}}-1}+\frac{e^{-\frac{\mu}{\rho}}\psi^{*,{\text{Col}}}-1}{2e^{-\frac{\mu}{\rho}}\psi^{*,{\text{Col}}}-1}(2(1-e^{-\frac{\mu}{\rho}}\psi^{*,{\text{Col}}}))^K\right).
\end{equation}
And the optimal backoff window sizes with the capture model can be obtained by combining (\ref{trans}) and (\ref{qm_cap}-\ref{temp2})  in the paper as
\begin{equation}
W^{*,\text{Cap}}_i=\left\{\!\!\!
\begin{array}{ll}
\hat{W}^{\text{Cap}}_0 2^i\,\,\,\,\,\,\,\,\,\,\,\,\,\text{if}\,\,\,\,\,\, \mu\ge \mu^{}_0\\
1 \,\,\,\,\,\,\,\,\,\,\,\,\,\,\,\,\,\,\,\,\,\,\,\,\,\,\,\,\,\,\,\text{otherwise,}\,\,\,\,\,\,
\end{array}\right.\notag
\end{equation}
$i=0,\ldots,K$, where
\begin{equation}\label{temp4}
\hat{W}^{\text{Cap}}_0=\frac{-2n}{\ln \psi^{*,{\text{Cap}}}}\left( \frac{e^{-\frac{\mu}{\rho}}(\psi^{*,{\text{Cap}}})^{\frac{\mu}{1+\mu}}}{2e^{-\frac{\mu}{\rho}}(\psi^{*,{\text{Cap}}})^{\frac{\mu}{1+\mu}}-1}+\frac{e^{-\frac{\mu}{\rho}}(\psi^{*,{\text{Cap}}})^{\frac{\mu}{1+\mu}}-1}{2e^{-\frac{\mu}{\rho}}(\psi^{*,{\text{Cap}}})^{\frac{\mu}{1+\mu}}-1}(2(1-e^{-\frac{\mu}{\rho}}(\psi^{*,{\text{Cap}}})^{\frac{\mu}{1+\mu}}))^K\right).
\end{equation}

\subsection{Simulation Results and Discussions}
In the following, we will validate the above results by using the ns-2 simulator. Here the simulation is based on the dei80211mr library. The dei80211mr library provides enhanced functionality such as the capture model based on the 802.11 implementation included in ns release 2.29 \cite{dei80211mr}. The source code of the simulations can be found in Appendix \ref{app4}.


\begin{figure*}
\centerline{\subfloat[]{\includegraphics[width=2.4in,height=1.95in]{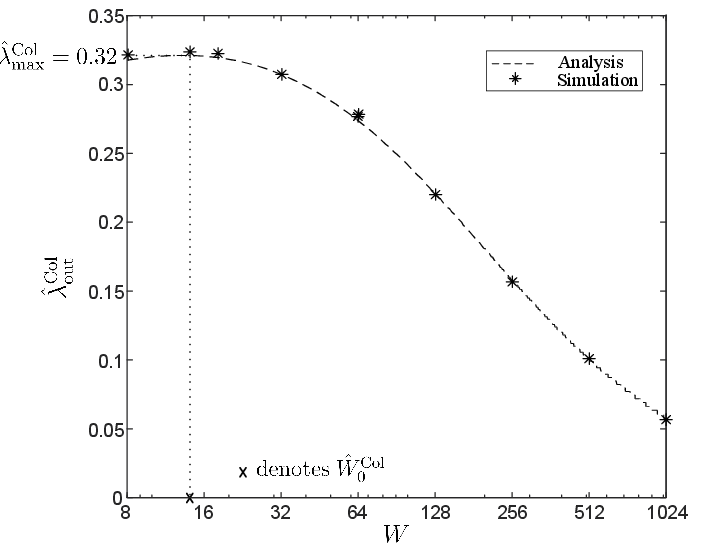}
\label{figure21}}\hfil
\subfloat[]{\includegraphics[width=2.4in,height=1.95in]{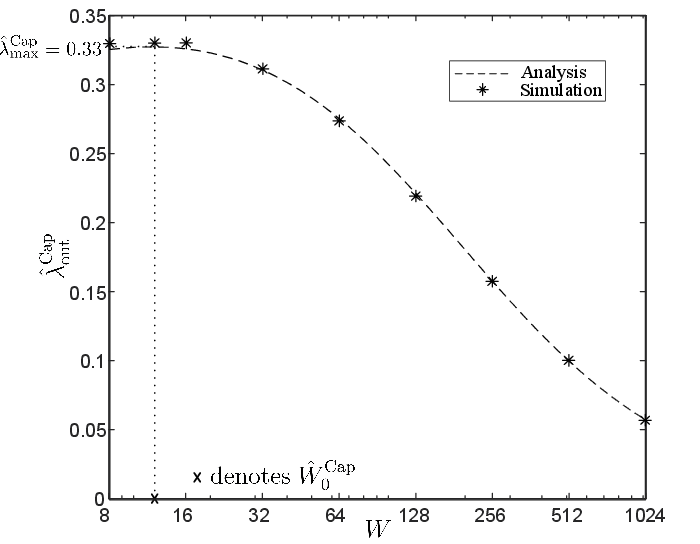}
\label{figure22}}\hfil
\subfloat[]{\includegraphics[width=2.4in,height=1.95in]{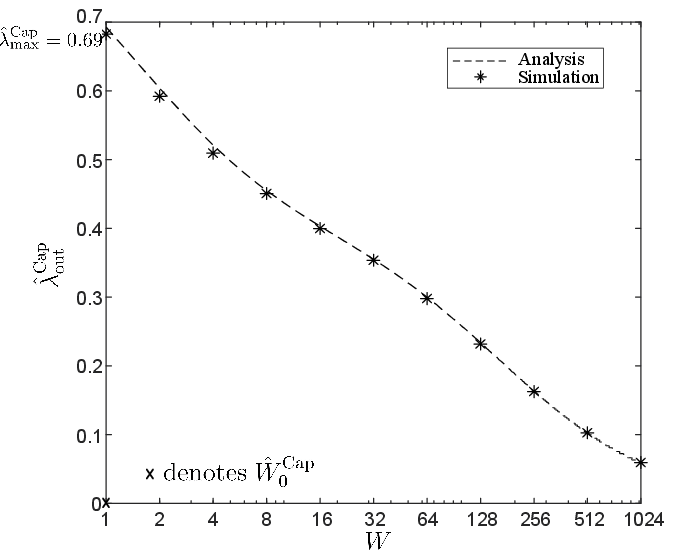}
\label{figure23}}} \caption{Network throughput $\hat{\lambda}^{}_{\text{out}}$ versus initial backoff window size $W$. $n=20$, $K=6$, $a=0.0247$, $x=34.36$ ($\tau_T= 40.44$ mini slots and $\tau_F= 34.36$ mini slots), and $W_i=W \cdot 2^i$. $\rho=10\text{dB}$, $\mu=10$. (a) Collision model. $\rho=10\text{dB}$, $\mu=10$. (b) Capture model. $\rho=10\text{dB}$, $\mu=10$. (c) Capture model. $\rho=-10\text{dB}$, $\mu=0.1$. }\label{figure2}
\end{figure*}

Fig. \ref{figure2} presents both the analytical and simulation results of the network throughput $\hat{\lambda}^{}_{\text{out}}$ versus the initial backoff window size $W$ for the collision and capture models. Note that in Fig. \ref{figure2}, 
$a$ and $x$ are obtained according to the values of system parameters in Table \ref{table1}, i.e., $a=0.0247$ and $x=34.36$. It can be seen from Fig. \ref{figure2} that the simulation results agree with the analysis well. For an IEEE 802.11 DCF network, both $\hat{\lambda}^{\text{Col}}_{\text{out}}$ and $\hat{\lambda}^{\text{Cap}}_{\text{out}}$ are sensitive to the initial backoff window size $W$. The corresponding optimal initial backoff window sizes have been derived as (\ref{temp3}) and (\ref{temp4}) for the collision model and the capture model, respectively, and are verified by simulation results presented in Fig. \ref{figure2}.

\appendices

\section{Derivation of \eqref{Eq7} and \eqref{pA_app}}\label{app0}

Based on the collision model, one packet can be successfully received if and only if there are no concurrent transmissions and its received SNR is above the receiver threshold $\mu$. We have
$p^{\text{Col}}=\Pr\{\text{no concurrent packet transmissions}\}\cdot\Pr\{\text{received SNR is above the threshold } \mu\}$.
According to the state transition process of each HOL packet shown in Fig. \ref{Chain}, the steady-state probability that one node attempts to access the channel given that the channel is sensed idle is given by $\sum _{i=0}^{K}\tilde{\pi }_{R_{i} } q_{i}$. Therefore, the probability that there are no concurrent transmissions is given by
$\Pr\{\text{no concurrent packet transmissions}\}=\left(1-\sum _{i=0}^{K}\tilde{\pi }_{R_{i} } q_{i}\right)^{n-1}$.
As the received SNR is exponentially distributed with mean $\rho$, the probability that the received SNR is above the receiver threshold $\mu$ can be written as
$\Pr\{\text{received SNR is above the threshold } \mu\}=\exp\left\{-\frac{\mu}{\rho}\right\}$.
As a result, we have
\begin{align} \label{Eq2}
p^{\text{Col}}&=\exp\left\{-\frac{\mu}{\rho}\right\}\cdot\left(1-\sum _{i=0}^{K}\tilde{\pi }_{R_{i} } q_{i}\right)^{n-1}
\mathop{\approx }\limits^{{\rm with \;large}\; n} \exp\left\{-\frac{\mu}{\rho}-n\sum _{i=0}^{K}\tilde{\pi }_{R_{i} } q_{i}
\right\}\notag\\
&{=}\exp\left\{-\frac{\mu}{\rho}\right\}\cdot \exp \left\{-\frac{a}{\alpha^{\text{Col}}
 p^{\text{Col}}} \cdot n\tilde{\pi }_{T }\right\},
\end{align}
where the probability of sensing the channel idle $\alpha^{\text{Col}}$ can be obtained as \eqref{Eq9}
by following a similar derivation to that in Appendix C of \cite{CSMA_Aloha}.
\eqref{Eq7} can then be obtained by substituting \eqref{sr} into \eqref{Eq2}, which has one single non-zero root $p^{\text{Col}}_A$ if $\{q_i\}_{i=0,\ldots,K}$ is a monotonic non-increasing sequence.



With the capture model, on the other hand, one packet can be successfully decoded as long as its received SINR is above the receiver threshold $\mu$. Specifically, for each transmitted packet from node $i$,
it can be successfully received if $\frac{|h_i|^2}{\sum_{j\in S_i}|h_j|^2+1/\rho}>\mu$
where $S_i$ is the set of nodes that transmit concurrently with node $i$.
With $|h_i|^2\sim\exp(1)$, the probability of successful transmission of one packet given $n_c$ other concurrent transmissions has been derived in \cite{Dua} as
$\exp\left\{-\frac{\mu}{\rho}\right\}\cdot\frac{1}{\left(1+\mu\right)^{n_c}}$.
The steady-state probability of successful transmission of HOL packets given that the channel is idle with the capture model, $p^{\text{Cap}}$, can then be obtained as
\begin{align} \label{ADD4}
p^{\text{Cap}}&{=}\sum_{n_c=0}^{n-1}
\exp\left\{-\frac{\mu}{\rho}\right\}\cdot\frac{1}{\left(1+\mu\right)^{n_c}}{\cdot}
\Pr\{n_c\text{ concurrent transmissions}\}.
\end{align}
As the probability that one node attempts to access the channel given that the channel is sensed idle is given by $\sum _{i=0}^{K}\tilde{\pi }_{R_{i} } q_{i}$, we have
\begin{align} \label{ADD5}
\Pr\{n_c\text{ concurrent transmissions}\}=\binom{n-1}{n_c}\left(\sum _{i=0}^{K}\tilde{\pi }_{R_{i} } q_{i}\right)^{n_c}\left(1-\sum _{i=0}^{K}\tilde{\pi }_{R_{i} } q_{i}\right)^{n-1-n_c}.
\end{align}
By combining \eqref{ls_def}, \eqref{ADD5} and \eqref{ADD4}, we have
\begin{align} \label{p_app_pre}
&p^{\text{Cap}}{=}\sum_{n_c=0}^{n-1}
\exp\left\{-\frac{\mu}{\rho}\right\}\cdot\frac{1}{\left(1+\mu\right)^{n_c}}{\cdot}\binom{n-1}{n_c}\left(\sum _{i=0}^{K}\tilde{\pi }_{R_{i} } q_{i}\right)^{n_c}\left(1-\sum _{i=0}^{K}\tilde{\pi }_{R_{i} } q_{i}\right)^{n-1-n_c}\mathop{\approx }\limits^{{\rm with \;large}\; n} \notag\\
& \exp\left\{{-}\frac{\mu}{\rho}\right\}{\cdot}\exp \left\{{-}\frac{n\mu}{1+\mu}\sum _{i=0}^{K}\tilde{\pi }_{R_{i} } q_{i}
\right\}{=}\exp\left\{{-}\frac{\mu}{\rho}\right\}{\cdot} \exp \left\{{-}\frac{\mu}{1{+}\mu}{\cdot}\frac{a}{\alpha^{\text{Cap}}
 p^{\text{Cap}}} {\cdot} n\tilde{\pi }_{T }\right\},
\end{align}
where the probability of sensing the channel idle $\alpha^{\text{Cap}}$ can be obtained as \eqref{alfa}
by following a similar derivation to that in Appendix C of \cite{CSMA_Aloha}. \eqref{pA_app} can then be obtained by substituting \eqref{sr} into \eqref{p_app_pre}, which has one single non-zero root $p^{\text{Cap}}_A$ if $\{q_i\}_{i=0,\ldots,K}$ is a monotonic non-increasing sequence.

\section{Proof of Theorem \ref{t1}}\label{app1}
\begin{proof}
Let
$\psi^{\text{Col}}=\exp\left\{\frac{\mu}{ \rho}\right\}\cdot p^{\text{Col}}_A$.
According to \eqref{Eq11}, we have
\begin{equation}\label{throu_appendix}
\hat{\lambda}^{\text{Col}}_{\text{out}}=
\frac{\frac{1}{ax}}{\exp\{\frac{\mu}{\rho}\}\cdot\frac{1+\frac{1}{x}-\psi^{\text{Col}}}{-\psi\ln\psi^{\text{Col}}}+\frac{1}{ax}-1}.
\end{equation}

According to \eqref{Eq7}, we have $p^{\text{Col}}_A<\exp\left\{-\frac{\mu}{\rho}\right\}$. Moreover, as $q_i\le 1$, ${i=0,\ldots,K}$, we have $p^{\text{Col}}_A\ge\exp\left\{-\frac{\mu}{\rho}-n\right\}$. As  a result, we have $\psi^{\text{Col}}\in[\exp\{-n\},1)$. According to \eqref{throu_appendix}, to maximize $\hat{\lambda}^{\text{Col}}_{\text{out}}$, we need to minimize $\frac{1+\frac{1}{x}-\psi^{\text{Col}}}{-\psi\ln\psi^{\text{Col}}}$ under the constraint that $\psi^{\text{Col}}\in[\exp\{-n\},1)$. It can be easily shown that $\frac{1+\frac{1}{x}-\psi^{\text{Col}}}{-\psi^{\text{Col}}\ln\psi^{\text{Col}}}$ is minimized when $\psi^{\text{Col}}=\psi^{*,{\text{Col}}}={-(1+1/x){\it {\mathbb W}}_{0}\left(-\tfrac{1}{e(1+1/x)}\right)}\in[\exp\{-n\},1)$. Therefore, the maximum throughput can be obtained by substituting $\psi^{\text{Col}}=\psi^{*,{\text{Col}}}$ into \eqref{throu_appendix}, and  the optimal transmission probability $q_i^{*,{\text{Col}}}$ can be obtained by combining $\psi^{\text{Col}}=\psi^{*,{\text{Col}}}$ and \eqref{Eq7}.
\end{proof}

\section{Proof of Theorem \ref{t3}}\label{app3}
\begin{proof}
Let
$\psi^{\text{Cap}}=\exp\left\{\frac{1+\mu}{ \rho}\right\}\cdot (p^{\text{Cap}}_A)^{\frac{1+\mu}{\mu}}$.
According to \eqref{th}, we have
\begin{align}\label{}
\hat{\lambda}^{\text{Cap}}_{\text{out}}{=}\frac{-\exp\left\{-\frac{\mu}{\rho}\right\}(\psi^{\text{Cap}})^{\frac{\mu}{1+\mu}}\ln \psi^{\text{Cap}}}{1{+}a{-}\psi^{\text{Cap}}{-}(1{-}ax)\sum_{i=1}^n\left(1{-}(\frac{1}{1{+}\mu})^{i-1}\exp\{{-}\frac{\mu}{\rho}\}\right)^i\binom{n}{i}\left({-}\frac{\ln\psi^{\text{Cap}}}{n}\right)^i\left(1{+}\frac{\ln\psi^{\text{Cap}}}{n}\right)^{n-i}}.
\end{align}
According to \eqref{pA_app}, we have $p^{\text{Cap}}_A\in\bigg[\exp\left\{-\frac{\mu}{\rho}\right\}\cdot \exp \left\{-\frac{n\mu}{1+\mu}\right\},\exp\left\{-\frac{\mu}{\rho}\right\}\bigg)$, which leads to $\psi^{\text{Cap}}\in[\exp\{-n\},1)$.
Therefore, the maximum throughput $\hat{\lambda}^{}_{\max}$ is given by
\begin{align}\label{}
&\hat{\lambda}^{\text{Cap}}_{\max}=\max_{\exp\{-n\}\le\psi^{\text{Cap}}<1}\hat{\lambda}^{\text{Cap}}_{\text{out}}
\notag\\
&{=}\max_{\exp\{-n\}\le\psi^{\text{Cap}}<1}\frac{-\exp\left\{-\frac{\mu}{\rho}\right\}(\psi^{\text{Cap}})^{\frac{\mu}{1+\mu}}\ln \psi^{\text{Cap}}}{1{+}a{-}\psi^{\text{Cap}}{-}(1{-}ax)\sum_{i=1}^n\left(1{-}(\frac{1}{1{+}\mu})^{i-1}\exp\{{-}\frac{\mu}{\rho}\}\right)^i\binom{n}{i}\left({-}\frac{\ln\psi^{\text{Cap}}}{n}\right)^i\left(1{+}\frac{\ln\psi^{\text{Cap}}}{n}\right)^{n-i}}.\notag
\end{align}

The first-order derivative of $\hat{\lambda}^{\text{Cap}}_{\text{out}}$ with respect to $\psi^{\text{Cap}}$ can be written as
\begin{align}\label{fir}
\frac{d\hat{\lambda}^{\text{Cap}}_{\text{out}}}{d\psi^{\text{Cap}}}=\frac{{-}\exp\left\{{-}\frac{\mu}{\rho}\right\}\cdot{(\psi^{\text{Cap}})^{{{-}\frac{1}{1+\mu}}}f(\psi^{\text{Cap}})}}{\left({1{+}a{-}\psi^{\text{Cap}}{-}(1{-}ax)\sum_{i=1}^n\left(1{-}\frac{\exp\{{-}\frac{\mu}{\rho}\}}{(1{+}\mu)^{i-1}}\right)^i\binom{n}{i}\left({-}\frac{\ln\psi^{\text{Cap}}}{n}\right)^i\left(1{+}\frac{\ln\psi^{\text{Cap}}}{n}\right)^{n-i}}\right)^2},
\end{align}
where \begin{align}
&f(\psi^{\text{Cap}})=\left(\frac{\mu}{1+\mu}\ln \psi^{\text{Cap}}+1\right)\cdot\left(1+a-\psi^{\text{Cap}}\right)+\psi^{\text{Cap}}\ln \psi^{\text{Cap}}-(1-ax)\sum_{i=1}^n\left(1-\frac{\exp\left\{-\frac{\mu}{\rho}\right\}}{(1+\mu)^{i-1}}\right)^i\notag\\
&\binom{n}{i}\left({-}\frac{\ln\psi^{\text{Cap}}}{n}\right)^i\left(1{+}\frac{\ln\psi^{\text{Cap}}}{n}\right)^{n-i-1}\left( \left(\frac{\mu}{1{+}\mu}\ln \psi^{\text{Cap}}{+}1\right)\left(1{+}\frac{\ln\psi^{\text{Cap}}}{n}\right){-}i{-}\ln\psi^{\text{Cap}}\right).
\end{align}

It can be easily obtained from \eqref{fir} that
\begin{align}\label{}
\frac{d\hat{\lambda}^{\text{Cap}}_{\text{out}}}{d\psi^{\text{Cap}}}\Big|_{\psi^{\text{Cap}}=1}={{-}\frac{1}{a}\exp\left\{{-}\frac{\mu}{\rho}\right\}}<0,
\end{align}
and
\begin{align}\label{}
\frac{d\hat{\lambda}^{\text{Cap}}_{\text{out}}}{d\psi^{\text{Cap}}}\Big|_{\psi^{\text{Cap}}=\exp\{-n\}}=\frac{{-}\exp\left\{{-}\frac{\mu}{\rho}+\frac{n}{1+\mu}\right\}f(\exp\{-n\})}{\left(1+a-\exp\{-n\}-(1-ax)\left(1-\frac{\exp\left\{-\frac{\mu_{}}{\rho}\right\}}{(1+\mu)^{n-1}}   \right)^n\right)^2}.
\end{align}
Let $\mu_0$ denote the root of $f(\exp\{-n\})=0$. When $\mu<\mu_0$, we have $\frac{d\hat{\lambda}^{\text{Cap}}_{\text{out}}}{d\psi^{\text{Cap}}}<0$ for $\psi^{\text{Cap}}\in[\exp\{-n\},1)$. Therefore, the maximum throughput $\hat{\lambda}^{\text{Cap}}_{\max}$ is achieved when $\psi^{\text{Cap}}=\exp\{-n\}$. The optimal transmission probability $q_i^{*,{\text{Cap}}}$ can be obtained by combining $\psi^{\text{Cap}}=\exp\{-n\}$ and \eqref{pA_app}. On the other hand, when $\mu>\mu_0$, we have $\frac{d\hat{\lambda}^{\text{Cap}}_{\text{out}}}{d\psi^{\text{Cap}}}>0$ for $\psi^{\text{Cap}}\in[\exp\{-n\},\psi^{*,{\text{Cap}}})$ and $\frac{d\hat{\lambda}^{\text{Cap}}_{\text{out}}}{d\psi}<0$ for $\psi^{\text{Cap}}\in(\psi^{*,{\text{Cap}}},1)$, where $\psi^{*,{\text{Cap}}}$ is the root of $f(\psi^{\text{Cap}})=0$. In this case, the maximum throughput $\hat{\lambda}^{\text{Cap}}_{\max}$ is achieved when $\psi=\psi^{*,{\text{Cap}}}$.  The optimal transmission probability $q_i^{*,{\text{Cap}}}$ can be obtained by combining $\psi^{\text{Cap}}=\psi^{*,{\text{Cap}}}$ and \eqref{pA_app}.
\end{proof}

\section{Source Codes of ns-2 Simulations}\label{app4}
\begin{lstlisting}
# ==============================================================
# Default Script Options
# ==============================================================
set opt(nn)      40; # Number of Nodes
set opt(pktsize)  2048;
# ==============================================================
# For topo pattern
# ==============================================================
set opt(TRlength)  100 ;#distance between transmitter and receiver
set opt(TTwidth)    0  ;#distance between transmitters
#===============================================================
set opt(RTSThreshold) 100000;#basic access is employed by default
set opt(CWMin)           31;
set opt(CWMax)         1023;
set opt(Time)          1000;
set sensingTreshdB     5   ;# sensing threshold in dB above noise power

proc usage {} {
    global argv0
    puts "\n usage: $argv0 \[-TRlength lenth\]\[-TTwidth width\]
    \[-RTSThreshold RTSThreshold\]\[-pktsize pktsize]
    \[-interval interval\]\[-CWMin CWMin\]\[-CWMax CWMax\]
    \[-Pt Pt]\[-noisePower noisePower]\[-ReceiverThreshold CPThresh]\n"
}

proc getopt {argc argv} {
	global opt
	for {set i 0} {$i < $argc} {incr i} {
		set arg [lindex $argv $i]
		if {[string range $arg 0 0] != "-"} continue
		set name [string range $arg 1 end]
		set opt($name) [lindex $argv [expr $i+1]]
	}
}
#usage
getopt $argc $argv

set val(chan)   Channel/WirelessChannel;
set val(prop)   Propagation/FreeSpace;
set val(netif)  Phy/WirelessPhy; # Rayleigh fading is added
set val(mac)    Mac/802_11;
set val(ifq)    Queue/DropTail/PriQueue;
set val(ll)     LL;
set val(ant)    Antenna/OmniAntenna;
set val(ifqlen) 5000;
set val(nn)     opt(nn);
set val(rp)     NOAH;

Mac/802_11 set ShortRetryLimit_   10000000       ;
Mac/802_11 set LongRetryLimit_    10000000       ;
Mac/802_11 set RTSThreshold_ $opt(RTSThreshold);

set opt(CSThresh) [expr $noisePower*pow(10,$sensingTreshdB/10.0)]
Phy/WirelessPhy set CSThresh_ $opt(CSThresh);
Phy/WirelessPhy set Pt_ $opt(Pt); # set transmisson power of each node
Phy/WirelessPhy set CPThresh_ $opt(CPThresh); # set receiver threshold
Phy/WirelessPhy set Noise_ $opt(noisePower); # set noise power
Phy/WirelessPhy set freq_ 914e+6;

Mac/802_11 set SlotTime_  0.000009;
Mac/802_11 set SIFS_      0.000016;
Mac/802_11 set CWMin_     $opt(CWMin);
Mac/802_11 set CWMax_     $opt(CWMax);
Mac/802_11 set dataRate_  65.0e6;
Mac/802_11 set basicRate_ 6.0e6;

# generate saturated Poisson traffic
Application/Traffic/Poisson set interval_ 0.0001
Application/Traffic/Poisson set packetSize_ 2048
Application/Traffic/Poisson set maxpkts_ 268435456

set ns_ [new Simulator]
set tracefd [open parrival.tr w]
$ns_ trace-all $tracefd
set topo [new Topography]
$topo load_flatgrid 300 300

create-god $val(nn)
   $ns_ node-config -adhocRouting $val(rp)\
                    -llType $val(ll)\
                    -macType $val(mac)\
                    -ifqType $val(ifq)\
                    -ifqLen $val(ifqlen)\
                    -antType $val(ant)\
                    -propType $val(prop)\
                    -phyType $val(netif)\
                    -channelType $val(chan)\
                    -topoInstance $topo\
                    -agentTrace ON\
                    -routerTrace OFF\
                    -macTrace ON\
                    -ifqTrace ON\
                    -movementTrace OFF

	for {set i 0} {$i < $val(nn) } {incr i} {
		set node_($i) [$ns_ node]	
		$node_($i) random-motion 0;# disable random motion
	}

    for {set i 0} {$i < $val(nn)/2 } {incr i} {
	set m  $i
	$node_($m) set X_ [expr $i* $opt(TTwidth)+20];
	$node_($m) set Y_ 20.0;
	$node_($m) set Z_ 0.0
    }
    for {set i 0} {$i < $val(nn)/2 } {incr i} {
	set m [expr $i+$val(nn)/2]
	$node_($m) set X_ [expr $i* $opt(TTwidth)+20];
 	$node_($m) set Y_ [expr 20+$opt(TRlength)];
	$node_($m) set Z_ 0.0
    }

exec ns r-poigen.tcl -pktsize $opt(pktsize) -nn $opt(nn)
-interval $opt(interval) -rate $opt(rate) >traffic
source traffic

set  Time1 [expr 0.01+$opt(Time)]
for {set i 0} {$i < $val(nn) } {incr i} {
    $ns_ at $opt(Time) "$node_($i) reset";
}
$ns_ at $opt(Time) "stop"
$ns_ at $Time1 "puts \"NS EXITING...\" ; $ns_ halt"
proc stop {} {
    global ns_ tracefd
    $ns_ flush-trace
    close $tracefd
}
puts "Starting Simulation..."
$ns_ run
\end{lstlisting}

%

\end{document}